\documentclass[a4paper,UKenglish,cleveref, autoref, thm-restate]{lipics-v2021}

\pdfoutput=1
\hideLIPIcs

\title{On Top-Down Pseudo-Boolean Model Counting}

\author{Suwei Yang}{Grabtaxi Holdings \and National University of Singapore}{}{0000-0003-1082-724X}{}

\author{Yong Lai}{Key Laboratory of Symbolic Computation and Knowledge Engineering of Ministry of Education, Jilin University, China}{}{0000-0002-6882-0107}{}

\author{Kuldeep S. Meel}{Georgia Institute of Technology \and University of Toronto}{}{0000-0001-9423-5270}{}

\authorrunning{S. Yang, Y. Lai and K.S. Meel}



\category{}

\relatedversion{}

\funding{This work was supported in part by the Grab-NUS AI Lab, a joint collaboration between GrabTaxi Holdings Pte. Ltd. and National University of Singapore, and the Industrial Postgraduate Program (Grant: S18-1198-IPP-II), funded by the Economic Development Board of Singapore. This work was supported in part by the Natural Sciences and Engineering Research Council of Canada (NSERC) [RGPIN-2024-05956], Jilin Provincial Natural Science Foundation [20240101378JC], and Jilin Provincial Education Department Research Project [JJKH20241286KJ].}

\acknowledgements{
The authors thank the reviewers for their feedback. The computational work for this article was performed on resources of the National Supercomputing Centre, Singapore.
}

\nolinenumbers %

\usepackage{amsmath}
\usepackage{amsfonts}
\usepackage{todonotes}
\usepackage{booktabs}
\usepackage{subcaption}
\usepackage{amsthm}
\usepackage{adjustbox}
\usepackage{multirow}
\usepackage{nicematrix}
\usepackage{scalefnt}
\usepackage{paralist}

\usepackage{tikz}
\usetikzlibrary{positioning, arrows, arrows.meta, trees, shapes.geometric}

\usepackage[ruled,linesnumbered,noend,noline]{algorithm2e}
\DontPrintSemicolon
\SetKwComment{Comment}{$\triangleright$\ }{}
\SetKwInOut{Input}{Input}
\SetKwInOut{Output}{Output}
\SetKwInput{KwCondONE}{Cond1}

\newcommand{\sharpsattd}{\ensuremath{\mathsf{SharpSAT}}-\ensuremath{\mathsf{TD}}}
\newcommand{\ganak}{\ensuremath{\mathsf{Ganak}}}
\newcommand{\approxmcpb}{\ensuremath{\mathsf{ApproxMCPB}}}
\newcommand{\approxasp}{\ensuremath{\mathsf{ApproxASP}}}
\newcommand{\sharpasp}{\ensuremath{\mathsf{sharpASP}}}

\newcommand{\pbcount}{\ensuremath{\mathsf{PBCount}}}
\newcommand{\pbmc}{\ensuremath{\mathsf{PBMC}}}
\newcommand{\heuname}{\ensuremath{\mathsf{VCIS}}}
\newcommand{\countalgname}{\ensuremath{\mathsf{CountPBMC}}}
\newcommand{\countrecalgname}{\ensuremath{\mathsf{Count}}}
\newcommand{\pbmcncs}{\ensuremath{\mathsf{PBMC}}-\ensuremath{\mathsf{NCS}}}
\newcommand{\pbcounter}{\ensuremath{\mathsf{PBCounter}}}

\newcommand{\dfour}{\ensuremath{\mathsf{D4}}}
\newcommand{\roundingsat}{\ensuremath{\mathsf{RoundingSat}}}
\newcommand{\gpmc}{\ensuremath{\mathsf{GPMC}}}
\newcommand{\addmc}{\ensuremath{\mathsf{ADDMC}}}
\newcommand{\dpmc}{\ensuremath{\mathsf{DPMC}}}

\newcommand{\apply}{\ensuremath{\mathsf{Apply}}}

\begin{document}

\maketitle

\begin{abstract}
Pseudo-Boolean model counting involves computing the number of satisfying assignments of a given pseudo-Boolean (PB) formula. In recent years, PB model counting has seen increased interest partly owing to the succinctness of PB formulas over typical propositional Boolean formulas in conjunctive normal form (CNF) at describing problem constraints. In particular, the research community has developed tools to tackle exact PB model counting. These recently developed counters follow one of the two existing major designs for model counters, namely the bottom-up model counter design. A natural question would be whether the other major design, the top-down model counter paradigm, would be effective at PB model counting, especially when the top-down design offered superior performance in CNF model counting literature.

In this work, we investigate the aforementioned top-down design for PB model counting and introduce the first exact top-down PB model counter, {\pbmc}. {\pbmc} is a top-down search-based counter for PB formulas, with a new variable decision heuristic that considers variable coefficients. Through our evaluations, we highlight the superior performance of {\pbmc} at PB model counting compared to the existing state-of-the-art counters {\pbcount}, {\pbcounter}, and {\ganak}. In particular, {\pbmc} could count for 1849 instances while the next-best competing method, {\pbcount}, could only count for 1773 instances, demonstrating the potential of a top-down PB counter design.
\end{abstract}

\section{Introduction}
\label{sec:intro}

Model counting refers to the task of computing the number of satisfying assignments of a given logical formula, typically in the form of a propositional Boolean formula in conjunctive normal form (CNF). Model counting has witnessed sustained attention from the research community in the past decades, owing to its relevance as a technique for tackling problems in various domains such as circuit vulnerability analysis, network reliability, and probabilistic inference~\cite{FSSC12,DMPV17,SBK05}. While model counting techniques are applied across various domains, a potential barrier to their wider adoption is the requirement for users to express the domain problem as a CNF formula before passing it to model counters. As a result, there has been recent interest in the research community in developing model counting tools for alternative input logic formats, such as {\pbcount}~\cite{YM24,YM25}, {\pbcounter}~\cite{LXY24}, and {\approxmcpb}~\cite{YM21} for pseudo-Boolean formulas and {\sharpasp}~\cite{KCM24} and {\approxasp}~\cite{KESH22} for answer set programming. The motivation for exploring alternative input formats is in part due to the fact that these inputs are more expressive and might be natural for certain application domains.

One notable emerging alternative input format for the model counting task is pseudo-Boolean (PB) formulas. PB formulas and related tools have seen an increase in attention from the research community in recent years, partly owing to the succinctness of PB formulas over typical CNF formulas~\cite{LMMW18}. There are typically two major paradigms when designing model counters -- bottom-up and top-down designs. Whilst both designs have been used in exact CNF model counters, the existing exact PB model counters, namely {\pbcount} and {\pbcounter}, are of the bottom-up design. In particular, both counters make use of algebraic decision diagrams (ADDs)~\cite{BFGHMPF93} as building blocks to perform counting in a bottom-up manner, starting with ADDs which represent individual constraints of the input PB formula and combining the individual ADDs in the computation process. The bottom-up design enables {\pbcount} and {\pbcounter} to reuse some of the decision diagram techniques found in existing ADD-based CNF model counters, namely {\addmc} and {\dpmc}~\cite{DPV20a,DPV20b}. A natural question one would wonder is \textit{`How would a top-down counter design perform in comparison to existing counters for PB model counting?'}.

In this work, we address the aforementioned question via a prototype top-down exact PB model counter, which we term {\pbmc}. {\pbmc} uses a top-down search-based approach along with techniques similar to conflict-driven clause learning (CDCL), adapted for PB formulas, to perform model counting. In particular, {\pbmc} builds on top of methodologies adapted from {\roundingsat} PB solver~\cite{EN18}, {\gpmc} CNF model counter~\cite{SHS17}, and {\pbcount} PB counter~\cite{YM24}. In addition, we developed a new variable decision ordering heuristic, {\heuname}, and a new component caching scheme tailored for PB formulas. In our evaluations, {\pbmc} demonstrates superior performance to state-of-the-art PB counters {\pbcount} and {\pbcounter}, and CNF counter {\ganak}~\cite{YM24,LXY24,SRSM19}. More specifically {\pbmc} is able to return counts for 1849 instances while {\pbcount} could only return for 1773 instances, {\pbcounter} for 1508 instances, and {\ganak} for 1164 instances. In addition, we show that our {\heuname} variable decision heuristic brings performance benefits, enabling {\pbmc} to improve from 1772 instances using prior heuristics to 1849 instances with {\heuname}.

\section{Background and Preliminaries}
\label{sec:background}

\subsection{Pseudo-Boolean Formula}
A pseudo-Boolean (PB) formula $F$ consists of a set $\Omega$ of one or more pseudo-Boolean constraints. Each PB constraint $\omega \in \Omega$ is of the form $\sum_{i=1}^{n} a_i x_i \triangleright k$, where $x_1, ..., x_n$ are Boolean literals, $a_1, ..., a_n$, and $k$ are integers, and $\triangleright$ is one of $\{\geq, =, \leq\}$. $a_1, ..., a_n$ are referred to as term coefficients in the PB constraint, where each term is of the form $a_i x_i$ and $k$ is known as the degree of a PB constraint. An assignment $\sigma$ is a satisfying assignment if it assigns values to all variables of $F$ such that all the PB constraints in $\Omega$ evaluate to \textit{true}.
PB model counting refers to the task of determining the number of satisfying assignments of $F$. 

Without loss of generality, the PB constraints in this work are all of the `$\geq$' type unless otherwise stated as `$\leq$' type constraints can be converted by multiplying -1 on both sides, and `$=$' type constraints can be replaced with a pair of `$\geq$' and `$\leq$' constraints. In addition, all coefficients are positive unless stated otherwise as the signs of coefficients can be changed using $\bar{x} = 1 - x$ and updating the degree accordingly. We use the term \textit{gap} $s$ to denote the satisfiability gap of a PB constraint $\omega$ under assignment $\sigma$. A gap value of 0 or less indicates $\omega$ is always satisfied under $\sigma$, otherwise $\omega$ is yet to be satisfied.

\begin{definition}
    Let $\omega$ be a PB constraint $\sum a_i \ell_i \geq k$ and $\sigma$ be a partial assignment. The gap $s$ of $\omega$ under $\sigma$ is given by $k - (\sum_{j : \sigma(\ell_j) = 1} a_j)$, where ${j : \sigma(\ell_j) = 1}$ indicates all literals $\ell$ that evaluates to true under $\sigma$.
\end{definition}

\subsection{Bottom-Up Model Counting}

Model counters that are of the bottom-up design typically make use of decision diagrams, which are directed acyclic graph (DAG) representations of functions, in the computation process. In particular, the counters make use of the {\apply} operation~\cite{B86} to build up the representation of the input logical formula in a bottom-up manner. Notable bottom-up counters including {\addmc}, {\dpmc}, {\pbcounter}, and {\pbcount}~\cite{DPV20a,DPV20b,DPM21,LXY24,YM24,YM25} employ algebraic decision diagrams (ADDs)~\cite{BFGHMPF93} to perform model counting tasks. These counters use early projection of variables to reduce the size of the intermediate decision diagrams in practice, as the complexity of the {\apply} operation is on the order of the product of sizes of the two operand diagrams. The aforementioned counters first construct individual ADDs of the input formula components i.e. clauses in the case of CNF and constraints in the case of PB formula. The individual ADDs are then merged via the {\apply} operation with the `$\times$' operator to produce intermediate ADDs. Variables are projected out using the {\apply} operation with the `$+$' operator according to either an initially computed plan or via heuristics. Once all ADDs are merged and all variables are projected out, the final ADD has a single leaf node that indicates the model count.

\subsection{Top-Down Model Counting}
Existing top-down model counters, typically for CNF formulas, adopt similar search methodologies in solvers, which tend to be the Conflict-Driven Clause Learning (CDCL) algorithm, to compute the model count. Notable counters with the top-down component caching design include {\dfour}, {\gpmc}, {\ganak}, and {\sharpsattd}~\cite{LM17,SHS17,SRSM19,KJ21} which demonstrated superior performance as winners of recent model counting competitions. On a high level, top-down counters iteratively assign values to variables until all variables are assigned a value or a conflict is encountered. In the process, sub-components are created by either branching on a variable or splitting a component into smaller variable-disjoint components. The counters cache the components and their respective counts to avoid duplicate computation. When all variables are assigned, the counter reaches a satisfying assignment and will assign the current component the count 1 and backtrack to account for other branches. When the counter encounters a conflict i.e. the existing assigned variable already falsifies the formula, it learns the reason for the conflict and prunes the search space such that this sub-branch of the search will be avoided in the remaining computation process. A leaf component in the implicit search tree either has a count of 0 for conflicts and a count of 1 when all variables are assigned. The count of a component that has variable-disjoint sub-components is the product of the counts of its sub-components. The count of a component that branches on a variable is the sum of the counts from the two branches. After accounting for all search branches, the counter returns the final count of satisfying assignments. 

\section{Approach}
\label{sec:approach}

In this section, we detail the general algorithm of our top-down PB model counter {\pbmc} as well as design decisions that enable a performant PB counter. 

\subsection{Search-Based Algorithm}

\begin{algorithm}
    \caption{{\countalgname} -- Counting Algorithm of {\pbmc}}
    \label{alg:pbmc-main}

    \Input{PB formula $F$}
    \Output{Model count}

    $F \gets \mathsf{Preprocess}(F)$ \;
    component $\varphi \gets F$; cache $\zeta \gets \emptyset$; assignment $\sigma \gets \emptyset$ \;

    $\mathsf{Count}(\varphi, \sigma, \zeta)$\;

    \Return $\zeta$[$\varphi$]\;

\end{algorithm}

The overall counting approach of {\pbmc} is shown in Algorithms~\ref{alg:pbmc-main} and~\ref{alg:pbmc-recur-helper}. {\countalgname} starts with preprocessing the input PB formula $F$ with techniques adapted from existing state-of-the-art counters~\cite{YM24}. Subsequently, {\countrecalgname} sets the preprocessed $F$ as the root component, and initializes an empty cache $\zeta$ and an empty assignment $\sigma$. {\countalgname} calls {\countrecalgname} to search the space of assignments and compute the model count. The main computation process of {\pbmc} in {\countrecalgname} comprises of the following steps -- propagation (line~\ref{line:pbmc-recur-helper:propagate}), conflict handling (lines~\ref{line:pbmc-recur-helper:propagate:conflict-handling-start}--\ref{line:pbmc-recur-helper:propagate:conflict-handling-end}), variable decision (lines~\ref{line:pbmc-recur-helper:var-dec-start}--\ref{line:pbmc-recur-helper:var-dec-end}), and component decomposition (line~\ref{line:pbmc-recur-helper:comp-decomp}). As shown in {\countrecalgname}, {\pbmc} starts with propagation so that variable decisions can be inferred and applied. If the propagation leads to a conflict, {\pbmc} would perform conflict analysis and backjumping using similar techniques as {\roundingsat}~\cite{EN18} (lines~\ref{line:pbmc-recur-helper:analysis}--\ref{line:pbmc-recur-helper:backjump}). If the conflict is a top-level conflict, that is it cannot be resolved by backjumping, {\pbmc} will terminate and return 0 as the formula is unsatisfiable (line~\ref{line:pbmc-recur-helper:propagate:conflict-handling-end}). Otherwise, {\pbmc} would learn a PB constraint that prunes the search space, backjump to the appropriate decision level $d$, and clear the components and recursive calls created after $d$. 

If there are no conflicts, {\pbmc} will attempt to split the current component into variable-disjoint child components, and the count of the current component would be the product of child component counts (lines~\ref{line:pbmc-recur-helper:comp-decomp}--\ref{line:pbmc-recur-helper:comp-count-prod}). There are two cases of component splitting -- (i) there are no new disjoint components and (ii) there are new disjoint components. In case (i) {\pbmc} will proceed to branch on another variable, and the count of the current component would be the sum of counts from the two branches (lines~\ref{line:pbmc-recur-helper:pick-lit}--\ref{line:pbmc-recur-helper:var-dec-end}). If there are no unassigned variables remaining, the count of the current component would be 1 (line~\ref{line:pbmc-recur-helper:sat-count}). Finally, after completing the implicit search tree exploration {\pbmc} returns the count of $F$ in cache $\zeta$. 

\begin{algorithm}
    \caption{{\countrecalgname} -- helper function of {\countalgname}}
    \label{alg:pbmc-recur-helper}

    \Input{component $\varphi$, assignment $\sigma$, cache $\zeta$}

    \lIf{\upshape $\varphi$ entry in $\zeta$}{\Return $\zeta[\varphi]$}
    isConf $\gets \mathsf{propagate}(\sigma)$\; \label{line:pbmc-recur-helper:propagate}
    \eIf{\upshape isConf} { \label{line:pbmc-recur-helper:propagate:conflict-handling-start}
        $\zeta[\varphi] \gets 0$\;
        success $\gets \mathsf{ConflictAnalysis(\sigma)}$ \; \label{line:pbmc-recur-helper:analysis}
        \lIf{\upshape success}{
            $\mathsf{Backjump}()$           \Comment*[f]{Clears related cache \& recursive calls} \label{line:pbmc-recur-helper:backjump}
        }
        \lElse{
            $\mathsf{ExitReturnZero}()$     \Comment*[f]{Unresolvable conflict, return 0}
        } \label{line:pbmc-recur-helper:propagate:conflict-handling-end}
    } {
        $\mathcal{C} \gets \mathsf{SplitComponent}(\varphi, \sigma)$ \; \label{line:pbmc-recur-helper:comp-decomp}
        \lIf{\upshape $|\mathcal{C}| > 1$}{$\zeta[\varphi] \gets \prod_{c \in \mathcal{C}} \mathsf{Count}(c, \sigma, \zeta)$} \label{line:pbmc-recur-helper:comp-count-prod}
        \ElseIf{\upshape $\varphi$ has unassigned variable}{ \label{line:pbmc-recur-helper:var-dec-start}
            $l \gets \mathsf{pickNextLit}(\varphi, \sigma)$\; \label{line:pbmc-recur-helper:pick-lit}

            $\zeta[\varphi] \gets \mathsf{Count}(\varphi \wedge l, \sigma \cup \{ l\}, \zeta) + \mathsf{Count}(\varphi \wedge \bar{l}, \sigma \cup \{ \bar{l} \} , \zeta)$ \; \label{line:pbmc-recur-helper:var-dec-end}
        }
        \lElse{
            $\zeta[\varphi] \gets 1$ \label{line:pbmc-recur-helper:sat-count}
        }
    }
    \Return $\zeta[\varphi]$
\end{algorithm}

\subsection{{\heuname} Variable Decision Heuristic}
\label{subsec:var-dec-heu}

In our top-down PB model counter {\pbmc}, we introduce a variable decision ordering heuristic that we term \textit{variable coefficient impact score} ({\heuname}). The intuition for a new variable decision heuristic for PB formulas came from the fact that the coefficient of each literal affects its impact on the overall counting process. Existing variable decision heuristics from CNF model counting literature do not have to consider the impact of coefficients, and therefore modifications are required to adapt to PB model counting. Additionally, tree decomposition may also be less effective when dealing with PB constraints that are relatively longer or constraints that share a lot of variables as in the case of multi-dimensional knapsack problems.

In our {\heuname} heuristics, we add a coefficient impact score for each variable as an equal-weighted additional component to prior variable decision heuristics adopted from {\gpmc}. The coefficient impact score for a variable $x$ is given by $\left( \sum_{j \in \Omega_x} b_x^j / k_j \right) \div |\Omega_x|$ where $\Omega_x$ is the set of constraints that variable $x$ appears in, $b_x^j$ is the coefficient of $x$ in constraint $j$, and $k_j$ is the degree of constraint $j$. At the start of the counting process, the score for each variable is computed, and variables are decided in descending order of scores in the counting process. In addition, the phase preferences of variable decisions are set such that the term coefficient with the largest coefficient impact score is applied, this is reflected in line~\ref{line:pbmc-recur-helper:pick-lit} of Algorithm~\ref{alg:pbmc-recur-helper}. The intuition is that branching on a variable with a larger score would have a greater impact on satisfiability gaps than that on a variable with a smaller score.

\subsection{Caching Scheme and Optimization} \label{subsec:cacheing-scheme}
After computing the count of each component, we store it in the cache to avoid redundant computations when the same component is encountered in a different branch of the counting process. A component in our caching scheme contains the following information for unique identification --- 
\begin{inparaenum}[1)]
\item the set of unassigned variables in the component,
\item the set of yet-to-be-satisfied PB constraints in the component, and
\item their current gaps. 
\end{inparaenum}
In our implementation, we stored variables and constraints of a component by their IDs. 

\begin{lemma}
    Given a PB formula $F$ and partial assignments $\sigma$ and $\sigma'$. Let component $c$ be that of $F|_{\sigma}$ and $c'$ be that of $F|_{\sigma'}$. If $\mathit{VarIDs}(c)= \mathit{VarIDs}(c')$, $\mathit{CstrIDs}(c)= \mathit{CstrIDs}(c')$, and each constraint $\omega$ in $c$ satisfies $\mathit{gap(\omega, \sigma)} = \mathit{gap(\omega, \sigma')}$, then $c = c'$.
\end{lemma}

\begin{proof}
    It is clear that $\mathit{Vars}(c)= \mathit{Vars}(c')$.
    For each constraint $\omega$ with $\mathit{CstrID}(\omega) \in \mathit{CstrIDs}(c)$, $\omega_{\sigma}$ and $\omega_{\sigma'}$ appear in $c$ and $c'$, respectively.
    As $\mathit{Vars}(c) = \mathit{Vars}(c')$, the left-hand side of $\omega_{\sigma}$ is the same as that of $\omega_{\sigma'}$. Then we know $\omega_{\sigma} = \omega_{\sigma'}$ as $\mathit{gap(\omega, \sigma)} = \mathit{gap(\omega, \sigma')}$.
\end{proof}

In our implementation, we cache the variable and constraint IDs in ascending order, thereby only needing to record the first ID and the differences between each $i^{th}$ and $i+1^{th}$. Since a difference is often smaller than an ID, we can use fewer bits to record each component.
Since each gap $s$ is positive, we can record it as $s - 1$. Additionally, we do not record gap for a clausal constraint, that is a constraint with coefficients and degree 1. Such a scheme is particularly efficient for the constraints with a lot of literals.

As a cache entry optimization trick to increase the cache hit rate, we perform saturation on the gaps during the generation of component cache IDs. The intuition for cache gap saturation comes from existing literature on PB solving where saturation is performed on the coefficients of a PB constraint~\cite{EN18} i.e. a term $a x$ in the PB constraint is modified to $k x$ where $k$ is the degree when $a > k$. Given a PB constraint $\sum_{j=1}^{m} a_j x_j \geq s$ where $s$ is the gap, our caching scheme changes $s$ in the cache ID to $\min{a_j}$ over all coefficients $a_j$ of unassigned variables when $0 < s < \min{a_j}$. This is sound as assigning any of the remaining literals to true would lead to the constraint being satisfied which is the same outcome without modification of $s$ i.e. the resulting gap being 0 or negative.

\section{Experiments}
\label{sec:experiments}

\subsection{Setup}

We implemented our top-down PB model counter, {\pbmc}, in C++, building on methodologies introduced in the {\gpmc} CNF model counter and {\roundingsat} PB solver. We performed comparisons to the existing state-of-the-art PB model counters {\pbcount}~\cite{YM24} and {\pbcounter}~\cite{LXY24}, and also model counting competition (MCC2024) winner {\ganak}~\cite{SRSM19}. For consistency, we employed the same benchmark suite as prior work~\cite{YM24}. The suite comprises of 3500 instances across three benchmark sets for PB model counting -- \textit{Auction} (1000 instances), \textit{M-dim Knapsack} (1000 instances), and \textit{Sensor placement} (1500 instances). We ran experiments on a cluster with AMD EPYC 7713 CPUs with 1 core and 16 GB memory for each instance and a timeout of 3600 seconds.

Through our evaluations, we seek to understand the performance of our top-down exact PB model counter {\pbmc}. In particular, we investigate the following research questions:
\begin{description}
    \item[RQ 1:] What is the performance of {\pbmc} compared to baselines?
    \item[RQ 2:] What is the performance impact of our {\heuname} heuristic in Section~\ref{subsec:var-dec-heu}?
    \item[RQ 3:] What is the performance impact of our cache entry optimization in Section~\ref{subsec:cacheing-scheme}?
\end{description}

\subsection{RQ 1: Performance of Top-down PB Counter {\pbmc}}

We compared {\pbmc} against the state-of-the-art PB model counter {\pbcount}, which has a bottom-up design, to understand how well a top-down design would perform. The results are shown in Table~\ref{tab:main-counting-results} and the cactus plot of runtime is shown in Figure~\ref{fig:pbmc-pbcount-runtime}.

\begin{figure*}[htb]
    \centering
    \begin{subfigure}{0.32\textwidth}
        \includegraphics[width=\linewidth]{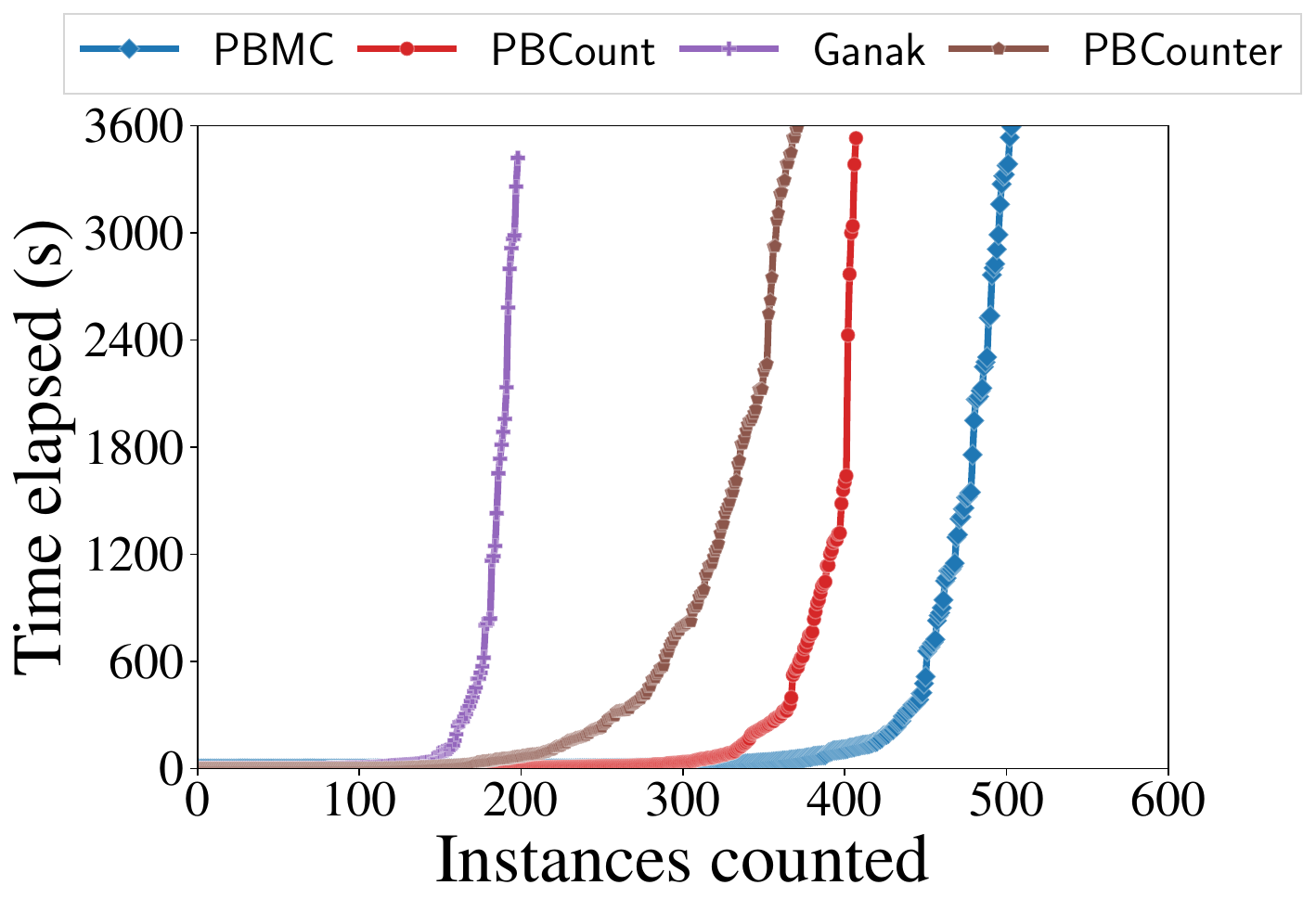}
        \caption{Auction}
    \end{subfigure}
    \hfill
    \begin{subfigure}{0.32\textwidth}
        \includegraphics[width=\linewidth]{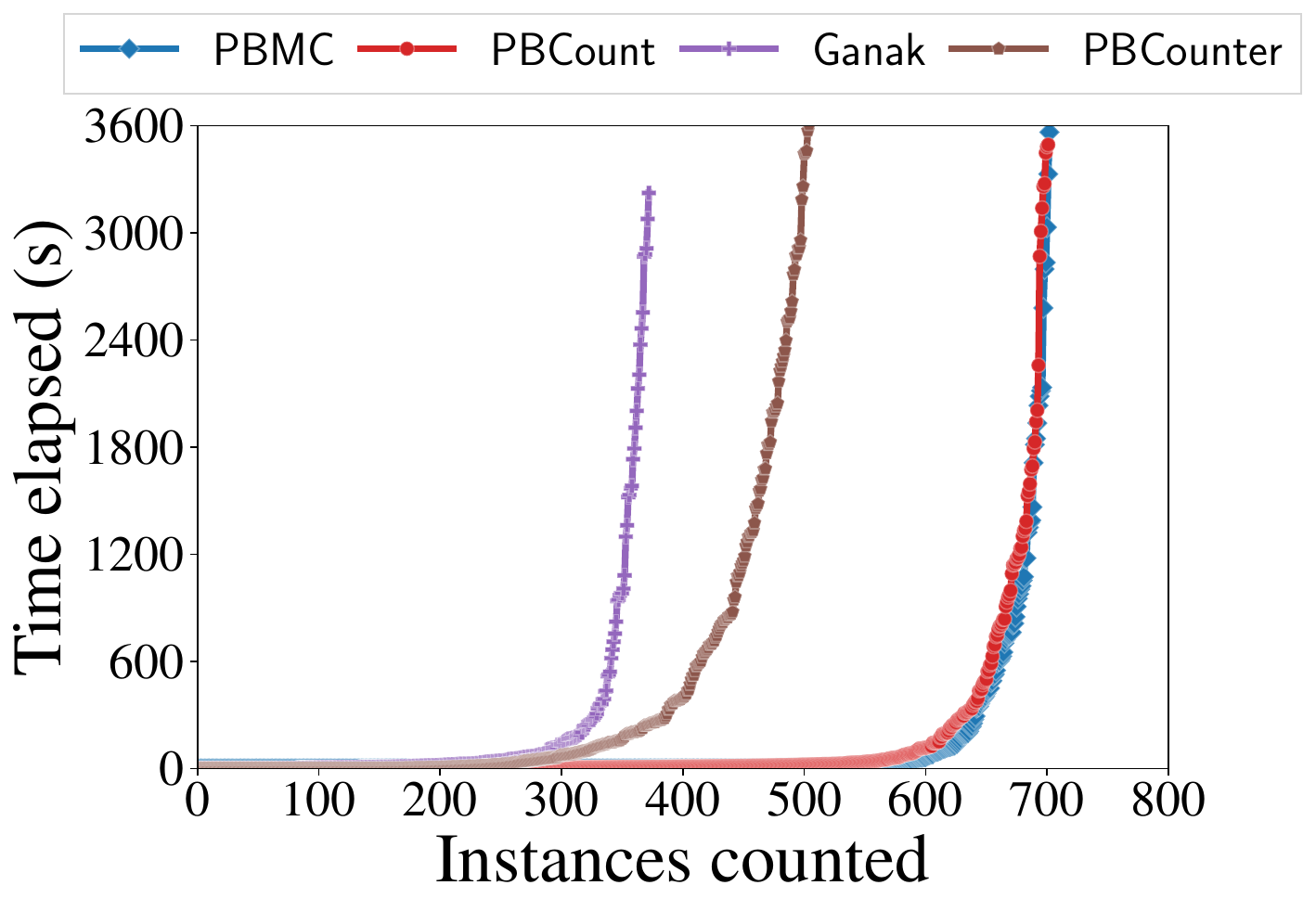}
        \caption{$\mathcal{M}$-dim Knapsack}
    \end{subfigure}
    \hfill
    \begin{subfigure}{0.32\textwidth}
        \includegraphics[width=\linewidth]{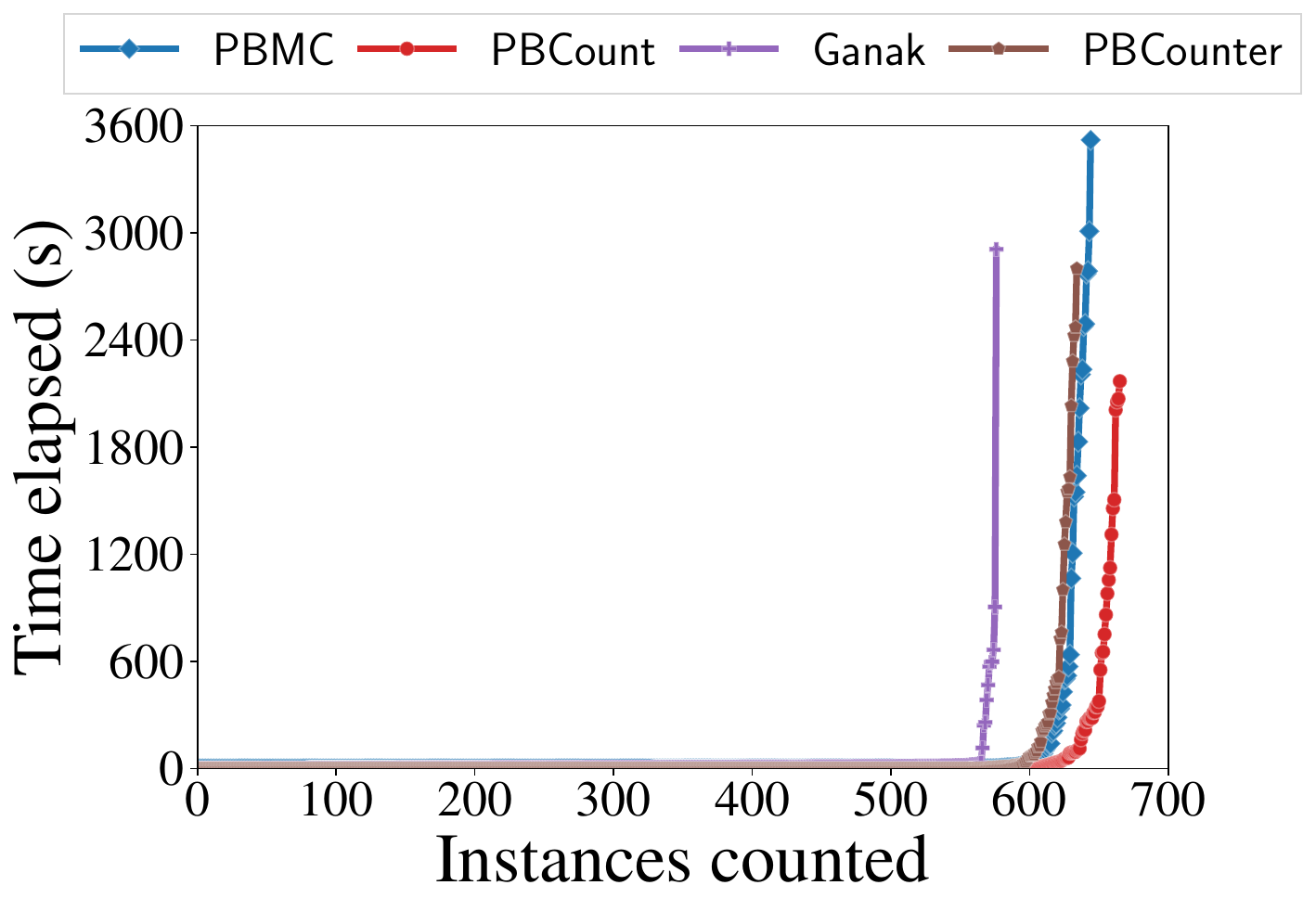}
        \caption{Sensor placement}
    \end{subfigure}
    \caption{Runtime cactus plots of {\ganak}, {\pbcount}, {\pbcounter}, and {\pbmc} with 3600s timeout. A point on the plot indicates the respective counter could return counts for $x$ instances in $y$ time.}
    \label{fig:pbmc-pbcount-runtime}
\end{figure*}

\begin{table}[h!]
    \centering
    \begin{small}
    \begin{NiceTabular}{l|r|r|r|r}
    \toprule
    Counters                        & Auction       & $\mathcal{M}$-dim knapsack    & Sensor placement  & Total         \\
    \midrule
    {\ganak}                        & 198           & 372                           & 576        & 1164        \\
    {\pbcount}                      & 407           & 701                           & \textbf{665}      & 1773          \\
    {\pbcounter}                    & 371           & 503                           & 634        & 1508         \\
    {\pbmc}                         & \textbf{503}  & \textbf{702}                  & 644               & \textbf{1849} \\
    \bottomrule
    \end{NiceTabular}
    \end{small}
    \caption{Number of benchmark instances counted by {\ganak}, {\pbcount}, {\pbcounter}, and {\pbmc} in 3600s, higher is better.}
    \label{tab:main-counting-results}
\end{table}

Overall, {\pbmc} is able to return counts for 1849 instances, 76 instances more than that of {\pbcount}, 341 instances more than {\pbcounter}, and 685 instances more than {\ganak}. More specifically, {\pbmc} leads in \textit{auction} and \textit{M dim-knapsack} benchmark sets and came in second in the \textit{sensor placement} benchmark set. In addition, there are 104 instances that could only be counted by {\pbmc}.

It is worth noting that the \textit{sensor placement} benchmark instances all have coefficients of 1 or -1, and that for each instance the degree of all but one PB constraint is 1.
We also compared {\pbmc} and {\pbcount} on a modified version of \textit{sensor placement}, whereby coefficients are not all 1 or -1 by accounting for potential redundancy requirements and varying costs associated with different sensor placement locations. On the cost-aware sensor placement instances {\pbmc} outperforms {\pbcount} by returning counts for 715 instances whereas {\pbcount} could only return for 678 instances.

\subsection{RQ 2: Performance Impact of {\heuname} Heuristics}

We conducted further experiments to understand the impact of our {\heuname} heuristic as introduced in Section~\ref{subsec:var-dec-heu}. In particular, we compared our implementation of {\pbmc} with {\heuname} against a baseline version of {\pbmc} that uses the existing variable decision heuristics from the {\gpmc} model counter. The results are shown in Table~\ref{tab:vheu-ablation-results} and the runtime cactus plot is shown in Figure~\ref{fig:pbmc-heu-runtime}.

\begin{figure*}[htb]
    \centering
    \begin{subfigure}{0.32\textwidth}
        \includegraphics[width=\linewidth]{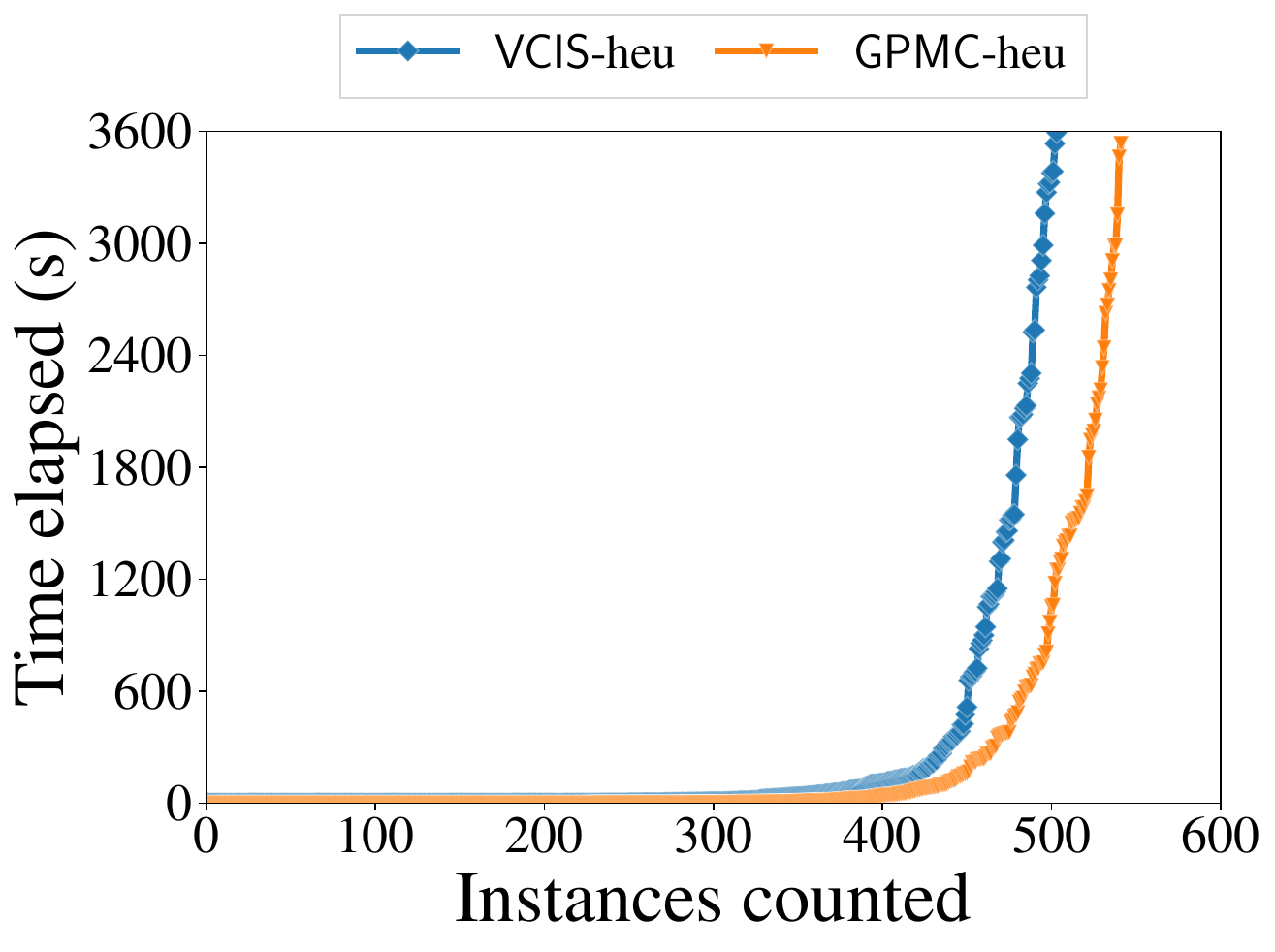}
        \caption{Auction}
    \end{subfigure}
    \hfill
    \begin{subfigure}{0.32\textwidth}
        \includegraphics[width=\linewidth]{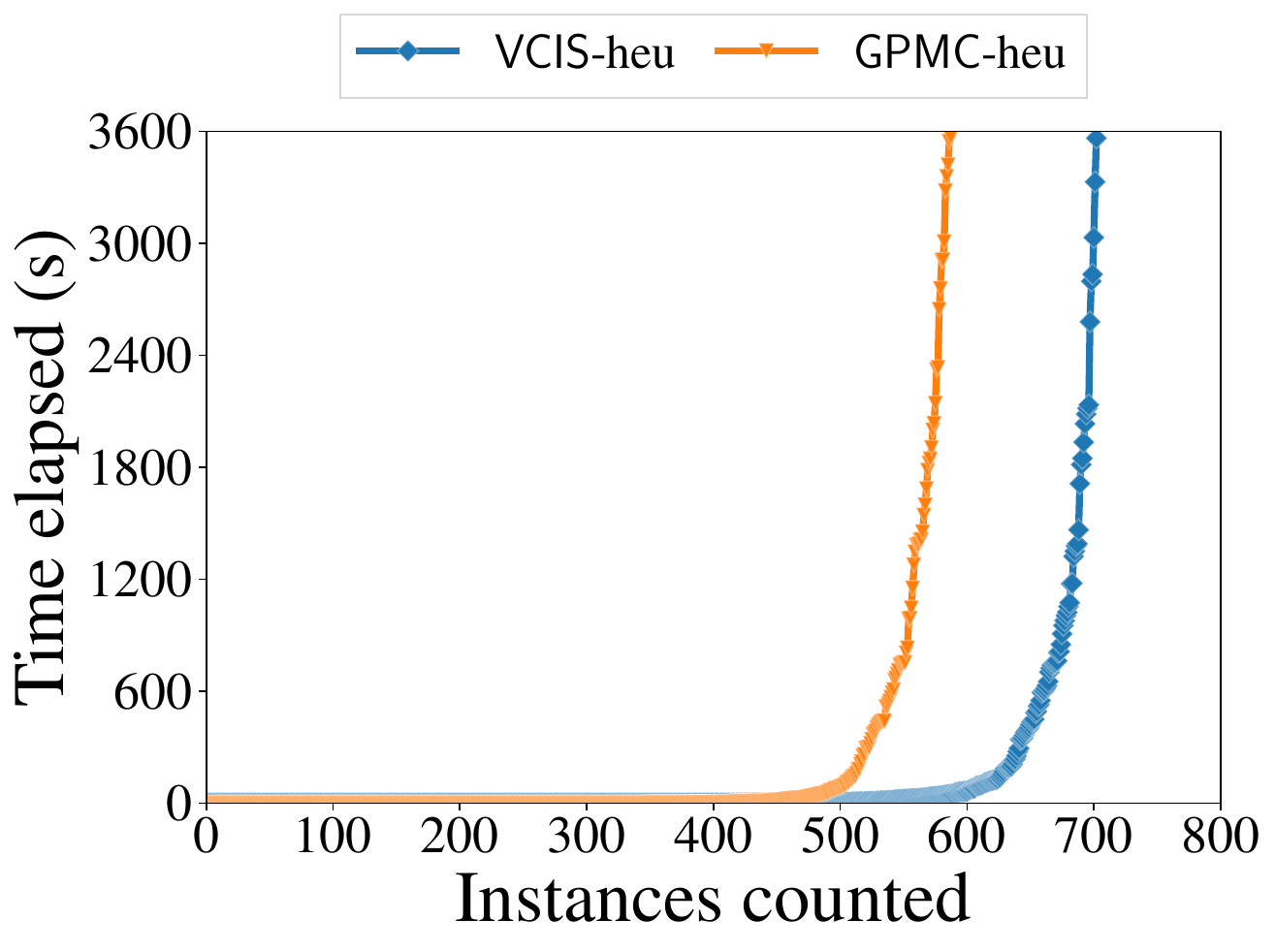}
        \caption{$\mathcal{M}$-dim Knapsack}
    \end{subfigure}
    \hfill
    \begin{subfigure}{0.32\textwidth}
        \includegraphics[width=\linewidth]{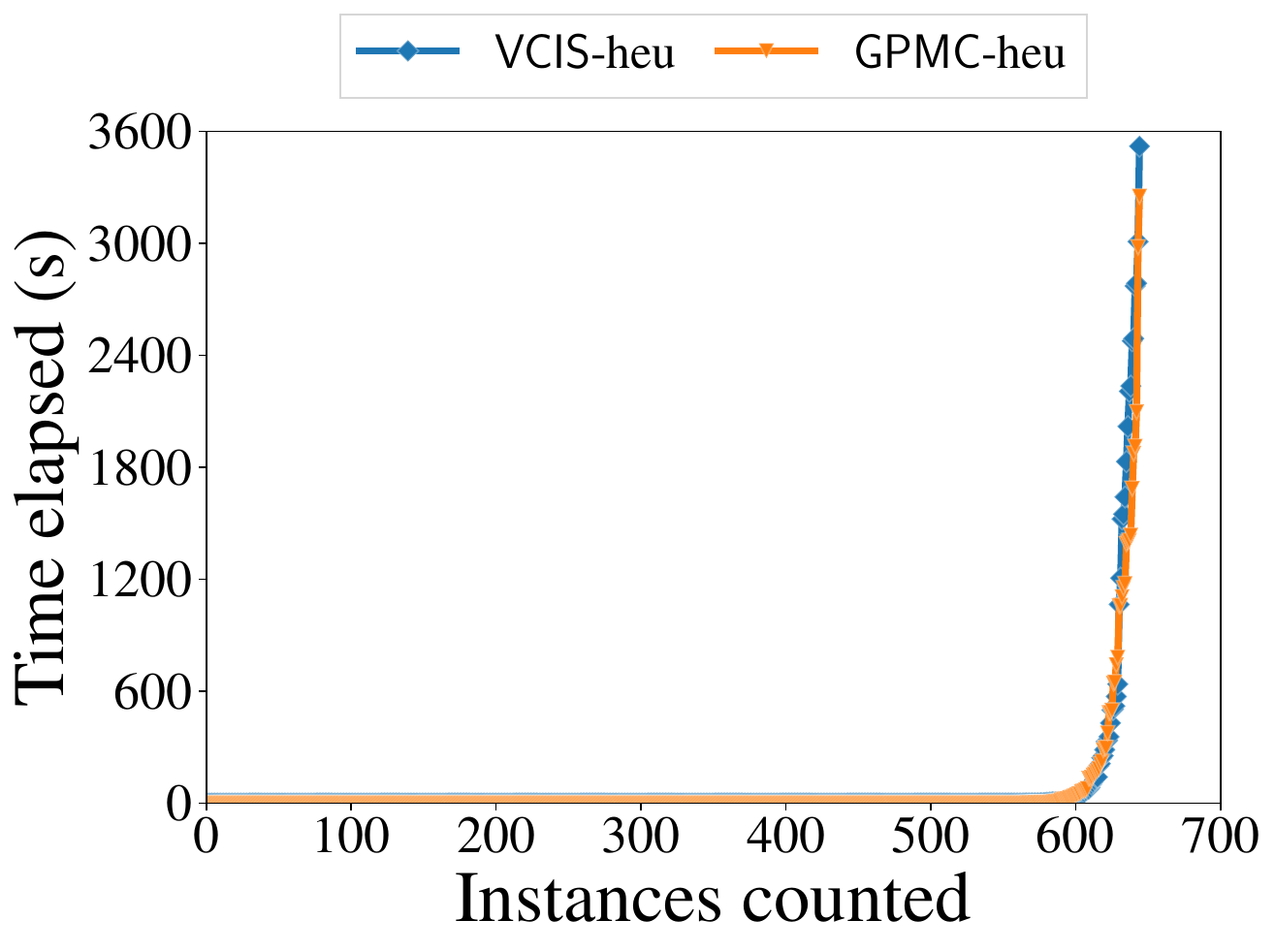}
        \caption{Sensor placement}
    \end{subfigure}
    \caption{Runtime cactus plots of {\pbmc} with {\heuname} heuristic and {\gpmc} variable decision heuristics ($\mathsf{PBMC}$-$\mathsf{GHeu}$) with 3600s timeout.}
    \label{fig:pbmc-heu-runtime}
\end{figure*}

\begin{table}[h!]
    \centering
    \begin{small}
    \begin{NiceTabular}{l|r|r|r|r}
    \toprule
    Counters                        & Auction       & $\mathcal{M}$-dim knapsack    & Sensor placement  & Total         \\
    \midrule
    {\gpmc}-heu                     & \textbf{541}  & 587                           & 644        & 1772          \\
    {\heuname}-heu                  & 503           & \textbf{702}                  & 644               & \textbf{1849} \\
    \bottomrule
    \end{NiceTabular}
    \end{small}
    \caption{Number of PB benchmark instances counted by {\pbmc}, using {\heuname} and {\gpmc} variable decision heuristics in 3600s.}
    \label{tab:vheu-ablation-results}
\end{table}

In the evaluations, {\pbmc} with our {\heuname} variable decision heuristic ({\heuname}-heu) is able to return counts for 1849 benchmark instances in total, substantially more than the 1772 benchmark instances when {\pbmc} is coupled with the existing heuristics from {\gpmc} ({\gpmc}-heu). In particular, {\heuname}-heu returned counts for 115 more instances than {\gpmc}-heu in \textit{$\mathcal{M}$-dim knapsack} benchmarks, while losing out 38 instances in \textit{Auction} benchmarks. Both heuristics performed the same in \textit{sensor placement} benchmarks. This set of evaluations highlights the tremendous impact of variable decision ordering in the context of PB model counting and demonstrates the performance benefits of our {\heuname} variable decision heuristic.

\subsection{RQ 3: Performance Impact of Cache Entry Optimization}
We conducted additional experiments to understand the impact of our optimization for our caching scheme mentioned in Section~\ref{subsec:cacheing-scheme}, that is modifying the gap value to the smallest coefficient of remaining unassigned variables. We compared {\pbmc} to a version with cache entry optimization disabled (denoted as {\pbmcncs}). We show the number of benchmark instances each configuration could count in Table~\ref{tab:cache-saturation-ablation-results}. The result shows that our cache entry optimization technique has a positive impact on the performance of {\pbmc}, although less significant than {\heuname} heuristics in the previous set of experiments.

\begin{table}[h!]
    \centering
    \begin{small}
    \begin{NiceTabular}{l|r|r|r|r}
    \toprule
    Counters                        & Auction       & $\mathcal{M}$-dim knapsack    & Sensor placement  & Total         \\
    \midrule
    {\pbmcncs}                      & 500           & 700                           & 644        & 1844          \\
    {\pbmc}                         & \textbf{503}  & \textbf{702}                  & 644        & \textbf{1849} \\
    \bottomrule
    \end{NiceTabular}
    \end{small}
    \caption{Number of PB benchmark instances counted by {\pbmc} with ({\pbmc}) and without ({\pbmcncs}) cache entry optimization in 3600s.}
    \label{tab:cache-saturation-ablation-results}
\end{table}

\subsection{Summary}

In \textbf{RQ 1}, {\pbmc} demonstrated a significant lead over state-of-the-art competitors. In the process, we observed that {\pbmc} performs better on PB instances with coefficients that are not just 1, as evident in the sensor placement benchmark set. In \textbf{RQ 2}, we showed the need for heuristics to consider PB-specific features such as coefficients through the significant performance improvement that {\heuname} enabled. We also observed that the top-down approach is rather sensitive to variable decision heuristics. Finally, in \textbf{RQ 3} we showed that by simply optimizing cache entries, one can get noticeable performance improvements.

\section{Conclusion and Future Work}
\label{sec:conclusion}

In this work, we explored the top-down design for model counters in the context of pseudo-Boolean model counting. In particular, we introduced our prototype PB model counter {\pbmc} which made use of methodologies from PB solvers and CNF counters as well as our {\heuname} heuristics to demonstrate leading performance over the existing state-of-the-art PB counter {\pbcount}. {\pbmc} highlights the promising performance of the top-down counter design in the context of PB model counting. It is worth noting that the current performance lead of top-down PB counter {\pbmc} over bottom-up PB counter (\pbcount) is not as significant as witnessed in the case of top-down vs bottom-up CNF counters. This could be perhaps partly due to the fact that top-down counters are rather sensitive to variable decision ordering, especially in the context of PB formulas where coefficients and degrees have an impact on overall performance, as shown in our evaluations. Therefore, it would be of interest to further study variable decision ordering heuristics in the context of PB formulas as future work, because existing heuristics from CNF model counting literature might not be ideal as they do not account for PB-specific features. In addition, it would also make sense to develop techniques to select the appropriate variable decision heuristic based on the given input. Another topic of interest would be preprocessing techniques for the top-down PB counter design, similar to CNF model counting where the top-down counters incorporated various preprocessing techniques which led to significant performance improvements over the years. We hope this study increases the visibility of PB model counting, inspires better counter designs, and more applications of PB model counting.

\bibliography{reference.bib}

\end{document}